\documentclass{article}

\usepackage{amsmath}
\usepackage{amsthm}
\usepackage{braket}
\usepackage{algorithm}
\usepackage{algpseudocode}
\usepackage{hyperref}
\hypersetup{
    colorlinks=true,
    linkcolor=blue,
    filecolor=magenta,
    urlcolor=blue,
    citecolor=blue,
}
\usepackage[backend=biber,style=alphabetic]{biblatex}
\theoremstyle{definition}

\newtheorem{theorem}{Theorem}
\newtheorem{lemma}{Lemma}

\newcommand{\MM}{\textsc{Maximum Matching} }
\newcommand{\pimin}{\pi_{\text{min}}}
\newcommand{\Ymin}{Y_1}

\algdef{SE}[DOWHILE]{Do}{DoWhile}{\algorithmicdo}[1]{\algorithmicwhile\ #1}

\title{Limitation of Quantum Walk Approach to the Maximum Matching Problem}
\author{Alcides~Gomes {Andrade J\'unior} and Akira Matsubayashi\footnote{%
The authors are with
 the Division of Electrical Engineering and Computer Science,
 Kanazawa University, Kanazawa, \mbox{920--1192} Japan.}
}
\date{}

\begin{document}

\maketitle

\begin{abstract}
The \MM problem has a quantum query complexity lower bound of
 $\Omega(n^{3/2})$ for graphs on $n$ vertices
 represented by an adjacency matrix.
The current best quantum algorithm has the query complexity $O(n^{7/4})$,
 which is an improvement over the trivial bound $O(n^2)$.
Constructing a quantum algorithm for this problem with a query complexity
 improving the upper bound $O(n^{7/4})$ is an open problem.
The quantum walk technique is a general framework for constructing
 quantum algorithms by transforming a classical random walk search
 into a quantum search,
 and has been successfully applied to constructing
 an algorithm with a tight query complexity for another problem.
In this work we show that the quantum walk technique fails to produce
 a fast algorithm improving the known (or even the trivial) upper bound
 on the query complexity.
Specifically, if a quantum walk algorithm
 designed with the known technique
 solves the \MM problem
 using $O(n^{2-\epsilon})$ queries with any constant $\epsilon>0$,
 and if the underlying classical random walk is independent of an input graph,
 then
 the guaranteed time complexity is larger than any polynomial of $n$.
\end{abstract}

\paragraph{Keywords}
quantum algorithm, query complexity, random walk, Markov chain, hitting time

\section{Introduction}
A matching of an undirected graph $G$ is defined as a subset of the edges
 of $G$, where no two edges share a vertex.
The \MM problem is to find a matching of a given graph $G$
 with the maximum possible number of edges.
This computational problem is one of the fundamental problems in graph theory
 and has many applications.
In this work we consider the time and query complexities
 of quantum algorithms to solve the \MM problem.

The \MM problem is a well studied problem whose first polynomial time algorithm
 was devised
 by Edmonds \cite{Edmonds1965}.
Subsequently, Hopcroft and Karp \cite{Hopcroft1973} presented
 an $O(\sqrt{n}m)$ time algorithm for bipartite graphs on $n$ vertices and
 $m$ edges.
This algorithm was then generalized to arbitrary graphs by
 Micali and Vazirani \cite{Micali1980}.
The $O(\sqrt{n}m)$ time complexity
 is met by another algorithm
 of Gabow \cite{Gabow2017}.
Another result to note is the randomized algorithm in \cite{Mucha2004},
 which exploits algebraic properties of the matching problem, and is able
 to achieve a run time of $O(n^{\omega})$, where $\omega$ is the exponent
 associated with the best known matrix multiplication algorithm.
The current best upper bound on $\omega$ is
 $2.371339$ \cite{ADVXXZ2025}.
For bipartite graphs, the time complexity
 was recently improved
 to $m^{1+o(1)}$ by Chen et al.\ \cite{Chen2022}, who presented
 an $m^{1+o(1)}$ time randomized algorithm for the maximum flows and
 the minimum-cost flows on directed graphs,
 to which the \MM problem on bipartite graphs can be reduced.

Concerning quantum algorithms for the \MM problem,
 D\"orn \cite{Dorn2009} presented an algorithm with time complexities
 $O(n^2\log^2 n)$ in the adjacency matrix model and $O(n\sqrt m\log^2 n)$
 in the adjacency list model.
This algorithm is obtained from the algorithm of Micali and Vazirani
 by transforming search procedures into quantum search procedures.
The technique applied to the transformation is Grover's algorithm,
 or its generalization, called amplitude amplification
 \cite{Grover1997,Brassard2002}.
This technique can be used to transform classical search algorithms into
 quantum search algorithms while obtaining a quadratic speedup.

Efficiency of an algorithm is measured by the query complexity as well as
 the time complexity.
The query complexity is a complexity measure that measures how much
 of the input an algorithm needs to access in order to produce its output.
More specifically, an algorithm is formulated as a model where accesses
 to the input, represented by a bit string, are made through
 a black-box function, which receives an argument $i$ and returns
 the $i$-th bit of the input.
The query complexity is then defined as the number of times that the
 algorithm makes calls to this black-box function.
The query complexity measure is useful in the study of quantum algorithms
 due to the fact that we have techniques \cite{Ambainis2002,Hoyer2007}
 that we can use to prove lower bounds for it.
This ability to set lower bounds permits us to set limits on how efficiently
 a problem can be solved in quantum computers (since the query complexity
 is always less than the time complexity), and thus how much of an advantage a
 quantum algorithm can possibly provide to a particular problem.

For the \MM problem, a lower bound of the quantum query complexity was
 established in \cite{Berzina2003,Zhang2004} to be $\Omega(n^{3/2})$
 in the adjacency matrix model.
For bipartite graphs, a quantum algorithm with a nearly optimal
 query complexity $O(n^{3/2}\log^2 n)$
 in the adjacency matrix model
 was
 achieved by Blikstad et al.\ \cite{Blikstad2022}.
This algorithm is obtained by designing
 a classical algorithm for the \MM problem on
 bipartite graphs with a nearly optimal number of OR-queries, which ask
 if a given set of pairs of vertices has
 at least one pair of adjacent vertices,
 and by transforming the OR-queries to quantum queries using
 Grover's algorithm.
The current best quantum algorithms for general graphs,
 proposed by Kimmel and Witter \cite{Kimmel2021},
 have
 query complexities $O(n^{7/4})$ in the adjacency matrix model
 and $O(n^{3/4}\sqrt{m+n})$ in the adjacency list model.
In the adjacency matrix model, the upper bound $O(n^{7/4})$
 is an improvement over the trivial bound $O(n^2)$.
These algorithms are obtained by transforming Gabow's algorithm
 into quantum versions using the technique of
 guessing decision trees \cite{Beigi2020},
 which transforms queries made by a classical algorithm, together with
 a guessing algorithm that predicts the results of the queries,
 into a reduced number of quantum queries.
Constructing a quantum algorithm to solve the \MM problem
 for general graphs
 with the query complexity
 improving $O(n^{7/4})$
 is an open problem.

Constructing quantum algorithms is a process that currently is still
 very complicated, and is mostly done through the use of techniques
 that can transform existing classical algorithms into quantum algorithms.
Examples of such techniques are
 Grover's algorithm (amplitude amplification) \cite{Grover1997,Brassard2002}
 and
 guessing decision trees \cite{Lin2016,Beigi2020},
 as applied to the \MM algorithms
 in \cite{Dorn2009,Blikstad2022} and
 in \cite{Lin2016,Beigi2020,Kimmel2021}, respectively.
Yet another technique that can be used to construct quantum algorithms is
 the quantum walk.
Given a Markov chain $P$ on a state space $X$ and an indicator function
 $\chi: X \to \{0,1\}$, we can construct a random walk algorithm
 to search for an element $x\in X$ with $\chi(x)=1$, called a marked element,
 by
 simulating the transitions of $P$ in $X$ until we reach a marked element.
The expected run time of the algorithm is proportional to the expected hitting
 time $\tau$ of the set $Y$ of marked elements, i.e.,
 the expected number of transitions to reach an element of $Y$
 for the first time.
The quantum walk technique is transformation of the classical random walk
 algorithm into a quantum search and able to provide a speedup
 of finding a marked element in expected time of order $\sqrt{\tau}$.
This technique was first established by Szegedy \cite{Szegedy2004} for
 the spacial case of $Y$ consisting of a single marked element.
Szegedy's work was generalized to any number of marked elements
 by Ambainis et al.\ \cite{Ambainis2020}.
Example applications of the quantum walk technique include
 an algorithm in \cite{Ambainis2007} that solves
 the \textsc{Element Distinctness}\footnote{%
Given as input a list of numbers $x_1, \dots, x_n$,
 the \textsc{Element Distinctness} problem involves
 determining whether or not there exist two distinct elements $x_i$ and $x_j$
 such that $x_i = x_j$.}
 problem using $O(n^{2/3})$ queries matching this problem's lower bound,
 and an algorithm in \cite{Magniez2007} that solves the
 \textsc{Triangle}\footnote{%
The \textsc{Triangle} problem consists of finding a triangle in an input graph.}
 problem using $\tilde{O}(n^{13/10})$ queries.

In this work we show that the quantum walk technique described above
 fails to produce a fast algorithm for the \MM problem
 improving the known (or even the trivial) upper bound on the query complexity.
Specifically, we prove that if a quantum walk algorithm
 for the \MM problem, designed using the known technique,
 has a query complexity of $O(n^{2-\epsilon})$
 in both the adjacency matrix and list models,
 where $\epsilon$ is any positive constant,
 and if the underlying Markov chain is independent of the edges of
 an input graph,
 then
 there exists an input graph with $\Theta(n^2)$ edges such that
 the algorithm needs an expected run time larger than any polynomial of $n$.
Moreover, we prove the existence of such an input graph that is bipartite.

We note that our assumption on the independence of the Markov chain
 specifically means that for each $n$, a common Markov chain
 (a common state set $X$ and a common transition matrix on $X$)
 is used for all graphs with $n$ vertices,
 regardless of their edge set.
We believe that this assumption is valid to a certain extent under the
 restricted, $O(n^{2-\epsilon})$, query complexity.
This is because, for an input graph with $\Theta(n^2)$ edges,
 no algorithm with the restricted query complexity can use
 the entire information of edges, and therefore,
 must setup a Markov chain (its state set and transition matrix)
 based only on a small part of edges, i.e., $O(n^{2-\epsilon})$ edges.
Although an algorithm may setup a Markov chain using this small part of edges,
 our assumption is viewed as an extreme case in the sense that
 the ratio of edges queried for setting up a Markov chain
 to the $\Theta(n^2)$ edges tends to $0$ as $n\rightarrow\infty$.
If the input graph is represented by an adjacency matrix, then
 because the adjacency matrix has the size of $n^2$ for any number of edges,
 the above discussion applies even to graphs with $o(n^2)$ edges.

We also note that our result shows a limitation of a quantum walk algorithm
 designed simply using the technique of \cite{Ambainis2020} adopted to
 a classical random walk.
There remain possibilities to overcome the presented limitation
 by sophisticated algorithms, such as quantum walks more adaptive to input graphs
 and/or combined with other techniques.

After describing some definitions in Sect.~\ref{sc:preliminaries},
 we prove our result on general graphs in Sect.~\ref{sc:limitation_general}.
In Sect.~\ref{sc:limitation_bipartite},
 we modify the proof in
 Sect.~\ref{sc:limitation_general} to our result on bipartite graphs.
We conclude the paper in Sect.~\ref{sc:conclusion}.

\section{Preliminaries}
\label{sc:preliminaries}
\subsection{Matchings}

A \emph{matching} of an undirected graph $G$ is a subset $M$ of edges in $G$
 such that
 no two edges in $M$ are adjacent, i.e., incident to a common vertex.
The matching $M$ is said to be \emph{maximum}
 if $G$ has no matching that contains more edges than $M$.
The problem of computing a maximum matching of a given graph is denoted by
 the \MM problem.

A matching $M$ of a graph $G$ is said to be \emph{perfect} if
 every vertex of $G$ is incident to an edge in $M$.
By definition any perfect matching is maximum.
We use in our proof
 two simple facts on the number of perfect matchings of a graph,
 Lemmas \ref{lemma:perfect-matching-count}
 and~\ref{lemma:subgraph-perfect-matching-count} below.
\begin{lemma}[\cite{Zaks1971}]
\label{lemma:perfect-matching-count}
Let $\Phi(n)$ denote the total number of distinct perfect matchings on
 a complete graph of $2n$ vertices.
Then $\Phi(n)=(2n-1)!!=(2n)!/(2^nn!)$.
\end{lemma}
\begin{lemma} \label{lemma:subgraph-perfect-matching-count}
A graph of $2n$ vertices and $m \geq n$ edges contains at most $m^n/n!$
 distinct perfect matchings.
\end{lemma}
\begin{proof}
A perfect matching contains $n$ edges chosen from $m$ possible edges.
Therefore, a loose upper bound on the number of perfect matchings is
 $\binom{m}{n}\leq m^n/n!$.
\end{proof}

\subsection{Query Complexity}
Let $A$ be an algorithm whose input space is $\{0,1\}^n$, for some
 positive integer $n$.
For each $x \in \{0,1\}^n$, let $Q_x$ be a black-box function,
 which receives an argument $i$ with $0\leq i\leq n-1$ and returns
 the $i$-th bit $x_i$ of $x$.
We consider a computational model where
 $A$ only has access to a given input $x\in\{0,1\}^n$ through
 the black-box function $Q_x$.
That is, every time $A$ needs to read a bit from $x$,
 it makes a call to
 $Q_x$ with the desired bit index and receives the corresponding bit
 returned by $Q_x$.
Each access to a bit of $x$ made through $Q_x$ is called a \emph{query}.
Denote by $q(x)$ the maximum number of queries that $A$ makes to compute
 the output for input $x$.
The \emph{query complexity}
 of $A$ is defined as the value $\max_{x \in \{0,1\}^n} q(x)$.

For a quantum algorithm, the equivalent
 formulation of a query consists of interpreting $Q_x$ as a unitary
 transformation whose action on $\ket{i}\ket{0}$ is defined
 as $Q_x\ket{i}\ket{0} = \ket{i}\ket{x_i}$.
If this transformation is called with the input $\ket{i}\ket{0}$,
 then it outputs the state $\ket{i}\ket{x_i}$,
 which contains the $i$-th bit of $x$.
The number of calls to $Q_x$ in the quantum algorithm determines
 the number of queries performed by the algorithm,
 just as defined in the classical case above.

\subsection{Markov Chains}
Let $X$ be a finite set of states, and let $\{S_n\}_{n=0}^{\infty}$ be a series
of random variables assuming values in $X$. 
The variable $S_n$ determines the state of a stochastic process at
 the $n$-th point in time (time here is discrete).
We consider the probability
 $\Pr\{S_{n} = a \mid S_{n-1} = b_{n-1}, S_{n-2} = b_{n-2}, \dots, S_0 = b_0 \}$,
 i.e., the probability of the variable $S_n=a$,
 assuming that $S_{i} = b_i$, for $i < n$.
The stochastic process involving the
variables $S_n$ is called a \emph{time-invariant Markov chain} provided
\[
\begin{split}
&\Pr\{S_{n} =a \mid S_{n-1} = b_{n-1}, S_{n-2} = b_{n-2}, \dots, S_0 = b_0\}\\
         =&\Pr\{S_{n} = a \mid S_{n-1} = b_{n-1}\}\\
         =&\Pr\{S_1 = a \mid S_0 = b_{n-1}\}.
\end{split}
\]
The time-invariant Markov chain can be represented by a matrix $P$ defined as
 $P_{ab} = \Pr\{S_1 = b \mid S_0 = a\}$.

Let $\pi$ be a probability distribution over the elements of $X$.
We interpret $\pi$ as a row vector, where the $x$-th component,
 denoted by $\pi_x$, is the probability of sampling $x$ from $\pi$.
The distribution $\pi$ is said to be \emph{stationary}
 if $\pi = \pi P$.
By this definition $\pi=\pi P^n$ for any $n>0$.
A Markov chain is said to be \emph{irreducible} if, for any $a,b\in X$,
 there exists some $n > 0$ such that $P^{n}_{ab} > 0$.
It is known that if a Markov chain
 (with a finite state set $X$ as introduced here)
 is irreducible, then there exists a unique
 stationary distribution $\pi$, and $\pi_x>0$ for all $x\in X$.
A Markov chain is said to be \emph{aperiodic} if, for any $x\in X$,
 the greatest common divider of all numbers $n$, such that $P^n_{xx}>0$, is $1$.
A Markov chain (with a finite state set)
 is said to be \emph{ergodic} if it is irreducible and aperiodic.
It is known that if a Markov chain is ergodic, then
 $\lim_{n\rightarrow\infty}P^n_{ab}=\pi_b$ for any $a,b\in X$.
A Markov chain is said to be \emph{reversible} if
 there exists a distribution $\pi$ such that $\pi_aP_{ab} = P_{ba}\pi_b$
 for all $a,b \in X$.
The distribution $\pi$ satisfying the reversibility condition is stationary.

Let $Y$ be a subset of the state space $X$.
The \textit{hitting time} of $Y$ is the random variable of
 the number of transitions to start from the first state,
 chosen according to an initial distribution,
 and to reach
 an element of $Y$ for the first time.

\section{Limitation of Quantum Walk Approach for General Graphs}
\label{sc:limitation_general}
\begin{algorithm}[t]
\caption{Random Walk Search}
\label{alg:random-walk}
\begin{algorithmic}[1]
    \Procedure{Random-Walk-Search}{$P$, $\chi$}
        \State Let $\pi$ be the stationary distribution of $P$
        \State Sample an initial state $x \in X$ according to $\pi$
        \While{$\chi(x) \neq 1$}
            \State
\begin{tabular}[t]{@{}l@{}}
 Let $P_x$ be the distribution for the transitions\\
 from $x$, i.e., the $x$-th row of $P$
\end{tabular}
            \State Sample $y$ according to $P_x$
            \State Set $x = y$
        \EndWhile
        \State \Return $x$
    \EndProcedure
\end{algorithmic}
\end{algorithm}

Let $P$ be (the matrix representation of)
 an ergodic Markov chain on a finite state space $X$.
Suppose we want to perform a search for an element $x$
 of $X$ satisfying $\chi(x) = 1$, where
 $\chi: X \to \{0,1\}$. A random walk search algorithm making use of $P$
 is given in Algorithm \ref{alg:random-walk}.
Basically, the algorithm chooses an initial state according to
 the stationary distribution,
 and simulates transitions of $P$ through the states
 of $X$ until it finds one state satisfying the search condition.
The expected
time until a target state is reached is given by the expected hitting time
of the set $\{x \in X \mid \chi(x) = 1\}$.
To determine the overall
 (time or query) complexity cost of the algorithm we must take into account the
cost of the operations:
\begin{enumerate}
\item (Setup)
The cost of setting up the stationary distribution $\pi$ of $P$,
 and sampling the initial state from $\pi$;
\item (Transition)
The cost of sampling from the distribution $P_x$,
 determined by the $x$-th row of $P$;
\item (Check)
 The cost of computing the function $\chi$.
\end{enumerate}

Suppose the costs needed for the Setup, Transition and Check
 operations are $S$, $T$ and $C$, respectively.
If $\tau$ is the expected hitting time of
 the set $\{x \in X \mid \chi(x) = 1\}$, then the expected total cost
 needed until a target element is found is given by $S + \tau(T + C)$.

In \cite{Ambainis2020}, for an ergodic and reversible Markov chain $P$,
 a perturbed Markov chain $P(s)$ with a parameter
 $s\in [0,1)$ is introduced, and then the quantum version of $P(s)$ is
 designed and analyzed.
The conclusion of \cite{Ambainis2020} is that the quantum version
 only needs a cost of $\tilde O(S+\sqrt{\tau}(T+C))$, where
 the parameters $S$, $T$, $C$, and $\tau$ are of $P$ (not of $P(s)$).
This provides a black-box strategy for obtaining quantum speedups from
 classical random walk algorithms, but is, unfortunately,
 not enough to produce a fast algorithm to find a maximum matching of
 an $n$-vertex graph using $O(n^{2-\epsilon})$ queries, as we show next.

Suppose a quantum walk algorithm for the \MM problem
 is transformed from a Markov chain that is
 independent of the edges of an input graph.
Through its execution, the algorithm executes the Setup, Transition and
 Check steps, as described in Algorithm~\ref{alg:random-walk}.
For each state $x$ of the Markov chain,
 let $\xi(x)$ denote the set of edges queried
 to compute $\chi(x)$ after the Setup step when the initial state is $x$.
A crucial point to understand is that, whenever the Check operation is
 executed in the state $x$, it must output the same result ($\chi(x)=1$ or $0$),
 independent of how many transitions have been performed before reaching $x$.
So the result of the Check step on the state $x$ must depend only
 on the set $\xi(x)$,
 because if $x$ were the initial state, then
 $\xi(x)$ would be the only information available about the input graph.
If the number of queries is limited to $O(n^{2-\epsilon})$, then
 $O(n^{2-\epsilon})$ queries are performed after the Setup step, and thus
 we see that $|\xi(x)| = O(n^{2-\epsilon})$.

The next theorem shows that this constraint forces the quantum walk
 algorithm to perform an excessively large number of
 transitions in order to find a perfect matching of a certain graph.
In the proof of the theorem,
 we first consider the situation that
 a complete graph of $2n$ vertices is input to the algorithm.
This means that
 the underlying Markov chain can be used to search for a state
 associated with a perfect matching on the complete graph.
We will show that there exists a perfect matching such that
 a super-polynomial expected number of transitions are needed
 in order to reach a state associated with the matching.
We then consider a graph having $2n$ vertices, $n^2$ edges, and
 this matching, as its unique perfect matching, is input to the algorithm.
Because the Markov chain is independent of the edges,
 it takes the same time to find the matching in this graph
 as in the complete graph, and thus
 a quantum quadratic speedup is not enough to achieve a polynomial time
 complexity.

\begin{theorem}
\label{th:limitation_general}
Suppose that a quantum walk algorithm for the \MM problem
 is transformed from
 a random walk search in Algorithm~\ref{alg:random-walk}, and that
 the underlying Markov chain is independent of the edges of an input graph.
If the query complexity of the quantum walk algorithm is $cn^{2-\epsilon}$
 in both the adjacency matrix and list models,
 where $n$ is the number of the vertices, and
 $c>0$ and $\epsilon>0$ are any constants,
 then
 there exists an input graph with $\Theta(n^2)$ edges
 such that the algorithm
 needs an expected run time larger than any polynomial of $n$.
\end{theorem}

\begin{proof}
Let $P$ be the underlying Markov chain
 of the quantum walk algorithm.
Since $P$ is ergodic and
 has a finite state $X$,
 there exists a stationary distribution $\pi$, such that
 $\pi_x>0$ for each $x\in X$.
Also, the random walk search starts from the stationary distribution
 as described in Algorithm~\ref{alg:random-walk}.

We first suppose that a complete graph of $2n$ vertices,
 denoted by $K_{2n}$,
 is input to the quantum walk algorithm.
Let $\{M_i\}_{i=1}^{\Phi(n)}$ be the collection of all
 distinct perfect matchings on $K_{2n}$, where $\Phi(n)$ is
 the number of such perfect matchings.
For each $1\leq i\leq\Phi(n)$, we define $Y_i$ as
 the set of states in $X$ associated with $M_i$, i.e.,
 $Y_i=\{x\in X\mid M_i\subseteq\xi(x)\}$, where $\xi(x)$ is the set of edges
 of $K_{2n}$ that are queried to compute $\chi(x)$ after the Setup step
 if the initial state is $x$.
The computation of $\chi(x)$ is limited to perform at most
 $cn^{2-\epsilon}$ queries for any $x \in X$;
 therefore, each set $\xi(x)$ must contain at most $cn^{2-\epsilon}$
 distinct edges.
Without loss of generality, we assume that
 $Y_1$ has the minimum stationary probability to be hit, denoted by $\pimin$,
 over all
 $\{Y_i\}_{i=1}^{\Phi(n)}$,
 i.e.,
 $\pimin=\min_{1\leq i\leq\Phi(n)}\sum_{x\in Y_i}\pi_x=\sum_{x\in Y_1}\pi_x$.
We denote by $H$ the hitting time of $Y_1$.
The expected value $\tau$ of $H$ can be formulated as
\[
\begin{split}
\tau &=E[H]=\sum_{i=1}^{\infty}i\Pr\{H=i\}
=\sum_{i=1}^{\infty}\sum_{j=1}^i\Pr\{H=i\}\\
&=\sum_{j=1}^{\infty}\sum_{i=j}^{\infty}\Pr\{H=i\}
=\sum_{j=1}^{\infty}\left[1-\sum_{i=0}^{j-1}\Pr\{H=i\}\right].
\end{split}
\]
The probability $\Pr\{H=i\}$
 of hitting a state in $\Ymin$ at time $i$ for the first time
 is
 at most the probability of hitting a state in $\Ymin$ at time $i$
 (not necessarily for the first time), which is equal to
 $\sum_{x\in Y_1}(\pi P^i)_x=\sum_{x\in Y_1}\pi_x=\pimin$.
Moreover, $\sum_{i=0}^{j-1}\Pr\{H=i\}\leq 1$ obviously.
Therefore, we have
\begin{equation}
\begin{aligned}[b]
\label{eq:tau}
\tau
&\geq\sum_{j=1}^{\infty}\left[1-\min\left\{j\pimin,1\right\}\right]
\geq\sum_{j=1}^{\lfloor\pimin^{-1}\rfloor}\left[1-j\pimin\right]\\
&=\lfloor\pimin^{-1}\rfloor
 -\frac{\lfloor\pimin^{-1}\rfloor(\lfloor\pimin^{-1}\rfloor+1)}{2}\cdot\pimin\\
&\geq\frac{\lfloor\pimin^{-1}\rfloor-1}{2}.
\end{aligned}
\end{equation}

We upper bound the probability $\pimin$.
Let $\Psi(n)$ be the maximum number of perfect matchings associated with
 a state, i.e.,
 $\Psi(n)=\max_{x\in X}\sum_{i=1}^{\Phi(n)}\chi_{Y_i}(x)$, where
 $\chi_{Y_i}(x)$ is the indicator function that returns $1$ if
 $x\in Y_i$, $0$ otherwise.
The sum of probabilities of hitting a state in $Y_i$, over all $i$,
 is upper bounded as
\[
\sum_{i=1}^{\Phi(n)}\sum_{x\in Y_i}\pi_x
=\sum_{x\in X}\sum_{i=1}^{\Phi(n)}\chi_{Y_i}(x)\pi_x
\leq\sum_{x\in X}\Psi(n)\pi_x
=\Psi(n).
\]
Therefore, we have
\begin{equation}
\label{eq:pimin}
\pimin
=\min_{1\leq i\leq\Phi(n)}\sum_{x\in Y_i}\pi_x
\leq\frac{1}{\Phi(n)}\sum_{i=1}^{\Phi(n)}\sum_{x\in Y_i}\pi_x
\leq\frac{\Psi(n)}{\Phi(n)}.
\end{equation}
The numbers $\Phi(n)$ and $\Psi(n)$ are estimated as
 $\Phi(n)=(2n)!/(2^nn!)$
 by Lemma~\ref{lemma:perfect-matching-count}, and
 $\Psi(n)\leq (c(2n)^{2-\epsilon})^n/n!$
 by Lemma~\ref{lemma:subgraph-perfect-matching-count} and
 $|\xi(x)|\leq c(2n)^{2-\epsilon}$.
Therefore, it follows from (\ref{eq:tau}) and (\ref{eq:pimin}) that
\[
\begin{split}
\tau
&=\Omega\left(\pimin^{-1}\right)
=\Omega\left(\frac{\Phi(n)}{\Psi(n)}\right)
=\Omega\left(\frac{(2n)!/(2^nn!)}{(c(2n)^{2-\epsilon})^n/n!}\right)\\
&=\Omega\left(\frac{(2n)!}{(2^{3-\epsilon}cn^{2-\epsilon})^n}\right)\\
&=\Omega\left(\frac{\sqrt{n}(2n/e)^{2n}}{(2^{3-\epsilon}cn^{2-\epsilon})^n}\right)
 \qquad\text{\small (Stirling's approximation)}\\
&=\Omega\left(\sqrt{n}\left(\frac{n^\epsilon}{2^{1-\epsilon}ce^2}\right)^n\right),
\end{split}
\]
 which is larger than any polynomial of the number $2n$
 of vertices of $K_{2n}$.

Now we suppose a graph $G$ having $2n$ vertices,
 $n^2$ edges, and a unique perfect matching $M_1$
 is input to the quantum walk algorithm.
Such a graph $G$ can be obtained from $M_1$ by applying
 Corollary~1.6 in \cite{Lovasz1972}.
Since the underlying Markov chain is independent of the edges of
 an input graph,
 the expected hitting time $\tau$
 of $\Ymin$ is the same in $G$ as in $K_{2n}$.
The expected run time for $G$ with a quadratic quantum speed up,
 order of $\sqrt\tau$, is still larger than any polynomial of the number
 of vertices of $G$.
\end{proof}

\section{Limitation of Quantum Walk Approach for Bipartite Graphs}
\label{sc:limitation_bipartite}
In the last part of the proof of Theorem~\ref{th:limitation_general},
 we utilize the existence of a graph $G$ that has $2n$ vertices,
 $n^2$ edges, and a unique perfect matching $M_1$.
We can construct such a graph that is bipartite as follows.

\begin{lemma}
\label{lm:bipartite}
For any perfect matching $M_1$ on $2n$ vertices,
 there exists a bipartite graph $G$ having
 $2n$ vertices, $\Theta(n^2)$ edges, and a unique perfect matching $M_1$.
\end{lemma}

\begin{proof}
Suppose that we have two collections
 of $n$ vertices $u_1,\ldots,u_n$ and $v_1,\ldots,v_n$.
We may assume without loss of generality
 that $M_1$ is the edge set $\{(u_i,v_i)\mid 1\leq i\leq n\}$.
We define that $G$ is the bipartite graph obtained by joining
 $u_i$ and $v_j$ for each $1\leq i\leq n$ and $i\leq j\leq n$.
The graph $G$ has $\sum_{i=1}^n(n-i+1)=\Theta(n^2)$ edges.

If $M$ is any perfect matching of $G$, then
 because $u_n$ is adjacent only to $v_n$, 
 the edge $(u_n,v_n)$ is contained in $M$.
This means that all other edges incident to $v_n$ are not contained in $M$.
Since $(u_{n-1},v_n)$ is not contained in $M$,
 $M$ contains the only remaining edge $(u_{n-1},v_{n-1})$
 incident to $u_{n-1}$.
This means that all other edges incident to $v_{n-1}$ are not contained in $M$,
 and that $M$ contains the only remaining edge $(u_{n-2},v_{n-2})$
 incident to $u_{n-2}$.
Iterating this argument, we conclude that
 $G$ has a unique perfect matching $M=\{(u_i,v_i)\mid 1\leq i\leq n\}=M_1$.
\end{proof}

Replacing the graph $G$ used in the proof of Theorem~\ref{th:limitation_general}
 with the bipartite graph of Lemma~\ref{lm:bipartite},
 we have the following theorem.

\begin{theorem}
\label{th:limitation_bipartite}
Suppose that a quantum walk algorithm for the \MM problem
 is transformed from
 a random walk search in Algorithm~\ref{alg:random-walk}, and that
 the underlying Markov chain is independent of the edges of an input graph.
If the query complexity of the quantum walk algorithm is $cn^{2-\epsilon}$
 in both the adjacency matrix and list models,
 where $n$ is the number of the vertices, and
 $c>0$ and $\epsilon>0$ are any constants,
 then
 there exists a bipartite input graph with $\Theta(n^2)$ edges
 such that the algorithm
 needs an expected run time larger than any polynomial of $n$.
\end{theorem}

\section{Conclusion}
\label{sc:conclusion}

In this work we considered the use of the quantum walk technique to the
 construction of quantum algorithms for the \MM problem.
We showed that
 the simple use of
 this technique fails in producing a fast algorithm for the \MM problem
 achieving $O(n^{2-\epsilon})$ query complexity,
 even on bipartite graphs.
The problem of finding an algorithm for the \MM problem
 improving the known upper bound $O(n^{7/4})$ on the query complexity,
 or finding a better lower bound $\omega(n^{3/2})$ is still open.
An improved algorithm
 appears to rely on other techniques for constructing quantum algorithms.

\printbibliography
\end{document}